\documentclass[onecolumn, 12pt, draftcls]{IEEEtran}
\usepackage{amsfonts}
\usepackage{amsmath}
\usepackage{graphicx}
\usepackage{color}
\usepackage{multicol}
\usepackage{amssymb}
\usepackage{subfigure}
\usepackage{bm}
\usepackage{latexsym}
\usepackage{stfloats}

\begin{document}
\title{Low-Complexity Hybrid Precoding in Massive Multiuser MIMO Systems}
\author{
Le~Liang,~\IEEEmembership{Student Member,~IEEE},
Wei~Xu,~\IEEEmembership{Member,~IEEE},\\
~and Xiaodai Dong,~\IEEEmembership{Senior Member,~IEEE}
\thanks{
L. Liang and X. Dong are with the Department of Electrical and Computer Engineering, University
of Victoria, Victoria, BC V8W 3P6, Canada (email: liang@uvic.ca, xdong@ece.uvic.ca).
}
\thanks{
W. Xu is with the National Mobile Communications Research Lab., Southeast University, Nanjing 210096, China (email: wxu@seu.edu.cn).}
}

\maketitle

\newtheorem{lemma}{Lemma}
\newtheorem{theorem}{Theorem}
\newtheorem{remark}{Remark}

\begin{abstract}
Massive multiple-input multiple-output (MIMO) is envisioned to offer considerable capacity improvement, but at the cost of high complexity of the hardware. In this paper, we propose a low-complexity hybrid precoding scheme to approach the performance of the traditional baseband zero-forcing (ZF) precoding (referred to as full-complexity ZF), which is considered a virtually optimal linear precoding scheme in massive MIMO systems. The proposed hybrid precoding scheme, named phased-ZF (PZF), essentially applies phase-only control at the RF domain and then performs a low-dimensional baseband ZF precoding based on the effective channel seen from baseband. Heavily quantized RF phase control up to $2$ bits of precision is also considered and shown to incur very limited degradation. The proposed scheme is simulated in both ideal Rayleigh fading channels and sparsely scattered millimeter wave (mmWave) channels, both achieving highly desirable performance.
\end{abstract}


\section{Introduction}
Massive multiple-input multiple-output (MIMO) is known to achieve high capacity performance with simplified transmit precoding/receive combining design \cite{Marzetta}-\cite{massiveNextGeneration}. Most notably, simple linear precoding schemes, such as zero-forcing (ZF), are virtually optimal and comparable to nonlinear precoding like the capacity-achieving dirty paper coding (DPC) in massive MIMO systems \cite{scalingMIMO}. However, to exploit multiple antennas, the convention is to modify the amplitudes and phases of the complex symbols at the baseband and then upcovert the processed signal to around the carrier frequency after passing through digital-to-analog (D/A) converters, mixers, and power amplifiers (often referred to as the radio frequency (RF) chain). Outputs of the RF chain are then coupled with the antenna elements. In other words, each antenna element needs to be supported by a dedicated RF chain. This is in fact too expensive to be implemented in massive MIMO systems due to the large number of antenna elements.

On the other hand, cost-effective variable phase shifters are readily available with current circuitry technology, making it possible to apply high dimensional phase-only RF or analog processing \cite{Love:EqualGain}-\cite{Veen:analogBF}. Phase-only precoding is considered in \cite{Love:EqualGain}, \cite{Zheng:UniformPowerConstraint} to achieve full diversity order and near-optimal beamforming performance through iterative algorithms.
The limited baseband processing power can further be exploited to perform multi-stream signal processing as in \cite{Zhang:variableshift},
where both diversity and multiplexing transmissions of MIMO communications are addressed with less RF chains than antennas. \cite{Veen:analogBF} then takes into account more practical constraints such as only quantized phase control and finite-precision analog-to-digital (A/D) conversion. Works in \cite{Love:EqualGain}-\cite{Veen:analogBF}, however, do not consider the multiuser scenario and are not aimed to maximize the capacity performance in the large array regime.

In this paper, we consider the practical constraints of RF chains and propose to design the RF precoder by extracting the phases of the conjugate transpose of the aggregate downlink channel to harvest the large array gain in massive MIMO systems, inspired by \cite{Zhang:variableshift}. Low-dimensional baseband ZF precoding is then performed based on the equivalent channel obtained from the product of the RF precoder and the actual channel matrix. This hybrid precoding scheme, termed PZF, is shown to approach the performance of the virtually optimal yet practically infeasible full-complexity ZF precoding in a massive multiuser MIMO scenario.
Furthermore, hybrid baseband and RF precoding has been considered for millimeter wave (mmWave) communications in works \cite{Roh:Samsung}-\cite{Sayeed:MUmmWave}. They share the idea of capturing ``dominant" paths of mmWave channels using RF phase control and the RF processing is constrained, more or less, to choose from array response vectors. We will also show in the simulation the desirable performance of our proposed PZF scheme in mmWave channels.


\section{System Model}

\begin{figure}
\centering
\includegraphics[width = 0.48\textwidth, clip,keepaspectratio]{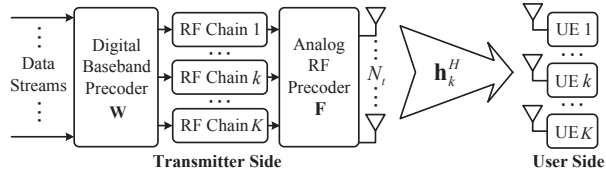}
\caption{System model of the hybrid mmWave precoding structure.} \label{fig:system}
\end{figure}

We consider the downlink communication of a massive multiuser MIMO system as shown in Fig.~\ref{fig:system}, where the base station (BS) is equipped with $N_t$ transmit antennas, but driven by a far smaller number of RF chains, namely, $K$. This chain limitation restricts the maximum number of transmitted streams to be $K$ and we assume scheduling exactly $K$ single-antenna users, each supporting single-stream transmission. As discussed, the downlink precoding is divided among baseband and RF processing, denoted by ${\bf W}$ of dimension $K \times K$ and ${\bf F}$ of dimension $N_t \times K$, respectively. Notably, both amplitude and phase modifications are feasible for the baseband precoder ${\bf W}$, but only phase changes can be made to the RF precoder ${\bf F}$ with variable phase shifters and combiners \cite{Zhang:variableshift}. Thus each entry of ${\bf F}$ is normalized to satisfy $|{\bf F}_{i, j}| = \frac{1}{\sqrt{N_t}}$ where $|{\bf F}_{i, j}|$ denotes the magnitude of the $(i, j)$th element of ${\bf F}$.

We adopt a narrowband flat fading channel and obtain the sampled baseband signal received at the $k$th user
\begin{align}\label{eq:yk}
y_k = {\bf h}_k^H {\bf F W s}+ n_k
\end{align}
where ${\bf h}_k^H$ is the downlink channel from the BS to the $k$th user, and ${\bf s}\in {\mathbb C}^{K \times 1}$ denotes the signal vector for a total of $K$ users, satisfying ${\mathbb E}[{\bf ss}^H] = \frac{P}{K} {\bf I}_K$ where $P$ is the transmit power at the BS and ${\mathbb E}[\cdot]$ is the expectation operator. To meet the total transmit power constraint, we further normalize ${\bf W}$ to satisfy $\|{\bf F W} \|_F^2 = K$.
$n_k$ denotes the additive noise, assumed to be circular symmetric Gaussian with unit variance, i.e., $n_k \sim \mathcal{CN}(0, 1)$.
Then the received signal-to-interference-plus-noise-ratio (SINR) at the $k$th user is given by
\begin{align}\label{eq:SINR}
\text{SINR}_k = \frac{ \frac{P}{K} |{\bf h}_k^H {\bf F w}_k|^2  }{1 + \sum \nolimits_{j \neq k} \frac{P}{K}|{\bf h}_k^H {\bf F w}_j|^2     }
\end{align}
where ${\bf w}_j$ denotes the $j$th column of ${\bf W}$.
If Gaussian inputs are used, the system can achieve a long-term average (over the fading distribution) spectral efficiency
\begin{align}\label{eq:rate}
R = \sum\nolimits_{k=1}^K{\mathbb E}\left[\log_2(1+{\text {SINR}_k}) \right].
\end{align}

\section{Hybrid Precoding in Massive MIMO Systems}\label{sec:scheme}
In massive MIMO systems, ZF precoding is known as a prominent linear precoding scheme to achieve virtually optimal capacity performance due to the asymptotic orthogonality of user channels in richly scattering environment \cite{scalingMIMO}. It is typically realized through baseband processing, requiring $N_t$ RF chains performing RF-baseband frequency translation and A/D conversion. This tremendous hardware requirement, however, restricts the array size from scaling large.

To alleviate the hardware constraints while realizing full potentials of massive multiuser MIMO systems, we propose to apply phase-only control to couple the $K$ RF chain outputs with $N_t$ transmit antennas, using cost-effective RF phase shifters. Low-dimensional multi-stream processing is then performed at the baseband to manage inter-user interference. The proposed low-complexity hybrid precoding scheme, termed phased-ZF (PZF), can approach the performance of the full-complexity ZF precoding, which is, as stated, practically infeasible due to the requirement of supporting each antenna with a dedicated RF chain. The spectral efficiency achieved by the proposed PZF scheme is then analyzed.

\subsection{Hybrid Precoder Design}\label{sec:scheme}
The structure shown in Fig.~\ref{fig:system} is exploited to perform the proposed hybrid baseband and RF joint processing, where the baseband precoder ${\bf W}$ modifies both the amplitudes and phases of incoming complex symbols and the RF precoder ${\bf F}$ controls phases of the upconverted RF signal. We propose to perform phase-only control at the RF domain by extracting phases of the conjugate transpose of the aggregate downlink channel from the BS to multiple users. This is to align the phases of channel elements and can thus harvest the large array gain provided by the excessive antennas in massive MIMO systems.
To clarify, denote ${\bf F}_{i,j}$ as the $(i, j)$th element of ${\bf F}$ and we perform the RF precoding according to
\begin{align}\label{eq:F}
{\bf F}_{i,j} = \frac{1}{\sqrt{N_t}} e^{j \varphi_{i,j}}
\end{align}
where $\varphi_{i,j}$ is the phase of the $(i, j)$th element of the conjugate transpose of the composite downlink channel, i.e., $\left[{\bf h}_1, \cdots, {\bf h}_K \right]$. Here we implicitly assume perfect channel knowledge at the BS which can potentially be obtained, e.g., through uplink channel estimation combined with channel reciprocity in time division duplex (TDD) systems \cite{Marzetta}. We note that efficient channel estimation techniques leveraging hybrid structures and rigorous treatment of frequency selectivity are an ongoing research topic of great practical interest.

Then at the baseband, we observe an equivalent channel ${\bf H}_{eq} = {\bf HF}$ of a low dimension $K\times K$ where ${\bf H} = \left[{\bf h}_1, \cdots, {\bf h}_K \right]^H$ is the composite downlink channel. Hence multi-stream baseband precoding can be applied to ${\bf H}_{eq}$, where simple low-dimensional ZF precoding is performed as
\begin{align}\label{eq:W}
{\bf W} =  {\bf H}_{eq}^H( {\bf H}_{eq} {\bf H}_{eq}^H  )^{-1}{\bf \Lambda}
\end{align}
where ${\bf \Lambda}$ is a diagonal matrix, introduced for column power normalization. With this PZF scheme, to support simultaneous transmission of $K$ streams, \emph{hardware complexity is substantially reduced, where only $K$ RF chains are needed, as compared to $N_t$ required by the full-complexity ZF precoding}.

{\textbf {Quantized RF Phase Control:}} According to \eqref{eq:F}, each entry of the RF precoder ${\bf F}$ differs only in phases which assume continuous values. However, in practical implementation, the phase of each entry tends to be heavily quantized due to practical constraints of variable phase shifters. Therefore, we need to investigate the performance of our proposed PZF precoding scheme in this realistic scenario, i.e., phases of the $KN_t$ entries of ${\bf F}$ are quantized up to $B$ bits of precision, each quantized to its nearest neighbor based on closest Euclidean distance. The phase of each entry of ${\bf F}$ can thus be written as ${\hat\varphi} = \left(2 \pi {\hat n} \right)/\left(2^B \right)$ where ${\hat n}$ is chosen according to
\begin{align}\label{eq:quantize}
{\hat n} = \arg \min \limits_{n \in \{0, \cdots, 2^B-1\}} \left| \varphi - \frac{2 \pi n}{2^B}  \right|
\end{align}
where $\varphi$ is the unquantized phase obtained from \eqref{eq:F}. Then the baseband precoder is computed by \eqref{eq:W} with the quantized ${\bf F}$.



\subsection{Spectral Efficiency Analysis in Rayleigh Fading Channels}
In this part, we analyze the spectral efficiency achieved by our proposed PZF and full-complexity ZF precoding in the limit of large transmit antenna size $N_t$ assuming Rayleigh fading. Closed-form expressions are derived, revealing the roles different parameters play in affecting system capacity.

Denoting the $k$th column of ${\bf F}$ by ${\bf f}_k$, we obtain
\begin{align}
y_k  = \left[{\bf h}_k^H {\bf f}_1, \cdots, {\bf h}_k^H {\bf f}_k, \cdots, {\bf h}_k^H {\bf f}_K  \right] {\bf W s} + n_k
\end{align}
based on \eqref{eq:yk}.
As described in Section \ref{sec:scheme}, ${\bf f}_k$ is designed by extracting the phases of ${\bf h}_k$, we thus have the diagonal term
\begin{align}
{\bf h}_k^H {\bf f}_k & = \frac{1}{\sqrt{N_t}} \sum\limits_{i = 1}^{N_t} |h_{i,k}|
\end{align}
where $h_{i,k}$ denotes the $i$th element of the vector ${\bf h}_k$.
Under the assumption that each element of ${\bf h}_k$ is independent and identically distributed (i.i.d.) complex Gaussian random variable with unit variance and zero mean, i.e., $h \sim {\mathcal {CN}}(0, 1)$, we conclude that $|h|$ follows Rayleigh distribution with mean $\frac{\sqrt{\pi}}{2}$ and variance $1-\frac{\pi}{4}$. When $N_t$ tends to infinity, the central limit theorem indicates
\begin{align}\label{eq:hkfk}
{\bf h}_k^H {\bf f}_k \mathop \sim {\mathcal N}(\frac{\sqrt{\pi N_t}}{2}, 1-\frac{\pi}{4}).
\end{align}

For the off-diagonal term, i.e., $j \neq k$, we have ${\bf h}_k^H {\bf f}_j = \frac{1}{\sqrt{N_t}} \sum\limits_{i=1}^{N_t} h_{i,k}^* {e^{j\varphi_{i,j}}}$, where $(\cdot)^*$ gives the complex conjugation. Its distribution is characterized in the lemma below.
\begin{lemma}\label{lemma}
In Rayleigh fading channels, the off-diagonal term ${\bf h}_k^H {\bf f}_j$ is distributed according to ${\bf h}_k^H {\bf f}_j \sim {\mathcal {CN}}(0,1)$.
\end{lemma}
\begin{proof}
The proof is achieved by analyzing the real and imaginary parts of ${\bf h}_k^H{\bf f}_j$ separately, followed by proving their independence. The proof is straightforward by definitions, and hence details are left out due to space limit.
\end{proof}

Based on Lemma~\ref{lemma}, we derive that the magnitude of the off-diagonal term, i.e., $|{\bf h}_k^H {\bf f}_j|$ follows the Rayleigh distribution with mean $\frac{\sqrt{\pi}}{2}$ and variance $1-\frac{\pi}{4}$. Compared with the diagonal term ${\bf h}_k^H {\bf f}_k$ given by \eqref{eq:hkfk}, it is safe to say \emph{the off-diagonal terms are negligible when the transmit antenna number $N_t$ is fairly large}. This implies that the inter-user interference is essentially negligible even without baseband processing at large $N_t$!
However, we note when $N_t$ assumes some medium high values, the residual interference may still deteriorate the system performance. Therefore we apply in our proposed scheme ZF processing at the baseband to suppress it as in \eqref{eq:W}.

We reason that even with ZF processing at the baseband, the spectral efficiency achieved is still less than it would be if the off-diagonal terms ${\bf h}_k^H{\bf f}_j$'s were precisely zero. In other words, the spectral efficiency achieved by PZF is upper bounded by $K{\mathcal R}$ with ${\mathcal R} = {\mathbb E}\left[\log_2\left(1 +\frac{P}{K} |{\bf h}_k^H{\bf f}_k|^2 \right)\right]$, which can be characterized by the following theorem using the limit equivalence type of argument \cite{Raghavan}.
\begin{theorem}\label{theorem}
The spectral efficiency achieved by the proposed low-complexity PZF precoding scheme is tightly upper bounded by $R_{\text{PZF}} \leq K{\mathcal R}$ where
\begin{align}\label{eq:rateHybrid}
\lim\limits_{N_t\rightarrow\infty} \frac{\mathcal R}{\log_2\left(1 + \frac{\pi}{4}\frac{P N_t}{K} \right)} = 1.
\end{align}
\end{theorem}

\begin{proof}
The per-user upper bound is derived as
\begin{align}
{\mathcal R} & =  {\mathbb E} \left[\log_2\left(1 + \frac{P}{K} \left(y + \frac{\sqrt{\pi N_t}}{2} \right)^2 \right) \right] \nonumber \\
& =  \log_2\left(1+\frac{\pi}{4}\frac{PN_t}{K} \right) + \underbrace{{\mathbb E}\left[\log_2\left(\frac{1+\frac{P}{K}\left( y + \frac{\sqrt{\pi N_t}}{2}\right)^2}{1+\frac{\pi N_t}{4}\frac{P}{K}}\right)\right]}\limits_{\Delta} \nonumber
\end{align}
where $y \sim {\mathcal N}(0, \sigma^2)$ with $\sigma=\sqrt{1-\frac{\pi}{4}}$. One may prove $\lim\limits_{N_t\rightarrow\infty}\Delta = 0$ by showing $\Delta$ is both upper and lower bounded by zero in the limit. An upper bound can be directly proved by applying the Jensen's Inequality. Proof of the lower bound is involved. Briefly, by defining $\rho \mathop = \limits^{\Delta} \frac{P}{K}$ and $a \mathop =\limits^{\Delta} \frac{\sqrt{\pi N_t}}{2}$, we have
\begin{align}
&\lim\limits_{N_t\rightarrow\infty}\Delta \geq \lim\limits_{a\rightarrow\infty} {\mathbb E}\left[\log_2\left(1+\frac{y}{a} \right)^2 \right] + \lim\limits_{a\rightarrow\infty}\log_2\frac{a^2}{\frac{1}{\rho} + a^2} \nonumber \\
& = \lim\limits_{a\rightarrow\infty}\frac{2a\log_2e}{\sqrt{2\pi}\sigma}\int_{0}^{+\infty} (\ln x)e^{-\frac{a^2(x-1)^2}{2\sigma^2}}\left(1 + e^{-\frac{2a^2x}{\sigma^2}} \right) dx \nonumber\\
& \mathop \geq\limits^{(a)} \lim\limits_{a\rightarrow\infty}\frac{2a \log_2e}{\sqrt{2\pi}\sigma}\left(1 + e^{-\frac{2a^2\xi}{\sigma^2}} \right)\int_0^1 (\ln x) e^{-\frac{a^2(x-1)^2}{2\sigma^2}}dx \nonumber \\
& = \lim\limits_{a\rightarrow\infty} \frac{2ae^{-\frac{a^2}{2\sigma^2}}}{\sqrt{2\pi}\sigma\ln2}\left(1 + e^{-\frac{2a^2\xi}{\sigma^2}} \right) \int_0^1 (\ln x) e^{\frac{a^2x}{2\sigma^2}} dx \nonumber\\
 & \mathop \geq\limits^{(b)}  \lim\limits_{a\rightarrow\infty} \frac{2ae^{-\frac{a^2}{2\sigma^2}}}{\sqrt{2\pi}\sigma\ln2}\left(1 + e^{-\frac{2a^2\xi}{\sigma^2}} \right)\frac{2\sigma^2 \left(1-e^{\frac{a^2}{2\sigma^2}}\right)}{a^2} =0 \nonumber
\end{align}
where (a) holds by shortening the integral range and then applying the Mean Value Theorem for Integral with $\xi\in(0,1)$. (b) is valid by using $\int_0^1(\ln x)e^{mx}dx \geq \frac{-e^m+1}{m}$ for $m\ge 0$.
\end{proof}

\begin{remark}
 Considering that the off-diagonal terms ${\bf h}_k^H {\bf f}_j$'s are essentially negligible when $N_t$ is large, we expect the derived closed-form upper bound to be very tight in the large antenna regime. This is further verified in the simulation results as shown in Fig.~2. Thus the closed-form upper bound serves as a good approximation of the spectral efficiency achieved by the proposed PZF precoding scheme at large $N_t$.
\end{remark}

The full-complexity ZF precoding vector (with unit norm) for the $k$th stream follows by projecting ${\bf h}_k$ onto the nullspace of ${\tilde {\bf H}}_k = \left[ {\bf h_1}, \cdots,{\bf h}_{k-1}, {\bf h}_{k+1}, \cdots, {\bf h}_K \right]^H$. In the spectral efficiency analysis, we exploit the property that users' channels are asymptotically orthogonal in massive multiuser MIMO systems \cite{Marzetta}. It indicates full-complexity ZF precoding converges to conjugate beamforming with \emph{inter-user interference forced to zero} , achieving
${\text {SINR}}_k \rightarrow \frac{P}{K}|{\bf h}_k|^2, \; \text{as} \; N_t \rightarrow\infty.$ Then according to \eqref{eq:rate}, we obtain the spectral efficiency of full-complexity ZF precoding in the limit of large $N_t$ as \cite{Goldsmith:closed}
\begin{align}\label{eq:rateFCZF}
R_{\text{FC-ZF}} \rightarrow & K{\mathbb E}\left[ \log_2\left(1 + \frac{P}{K} \left|{\bf h}_k\right|^2 \right)  \right] \nonumber \\
 = & Ke^{\frac{K}{P}} \log_2e \sum\limits_{n=1}^{N_t} E_n\left(\frac{K}{P} \right)
\end{align}
by acknowledging that $|{\bf h}_k|^2$ follows chi-squared distribution with $2N_t$ degrees of freedom and $E_n(x)$ is the exponential integral of order $n$.

\section{Simulation Results}

\subsection{Large Rayleigh Fading Channels}
We numerically compare our proposed PZF precoding scheme in Fig.~\ref{fig:mainCompareScheme} along with its quantized version against the full-complexity ZF scheme, which is deemed virtually optimal in the large array regime but practically infeasible due to the requirement of $N_t$ costly RF chains.
It is observed that the proposed PZF precoding performs measurably close to the full-complexity ZF precoding, with less than $1$ dB loss but substantially reduced complexity. As for the heavily quantized phase control, we find that with $B = 2$ bits of precision, i.e., phase control candidates of $\{0, \pm \frac{\pi}{2}, \pi \}$, the proposed scheme suffers negligible degradation, say less than $1$ dB.

The derived analytical spectral efficiency expressions \eqref{eq:rateHybrid} and \eqref{eq:rateFCZF} are also plotted in Fig.~\ref{fig:mainCompareScheme}. We observe that the derived closed-form expressions are quite accurate in characterizing spectral efficiencies achieved by the proposed PZF precoding and full-complexity ZF precoding schemes throughout the whole signal-to-noise (SNR)\footnote{Here $\text{SNR} = P$ is the common average SNR received at each antenna with noise variance normalized to unity.}range, thus providing useful guidelines in practical system designs.
\subsection{Large mmWave Multiuser Channels}
Apart from ideal i.i.d. Rayleigh fading channels, our proposed PZF scheme can also be applied to the mmWave communication which is known to have very limited multipath components.
To capture this poor scattering nature, in the simulation, we adopt a geometric channel model \cite{Roh:Samsung}-\cite{Sayeed:MUmmWave}
\begin{align}\label{eq:channel}
{\bf h}_k^H = \sqrt{\frac{N_t}{N_p}} \sum \limits_{l=1}^{N_p}\alpha_l^k {\bf a}^H\left(\phi_l^k, \theta_l^k \right)
\end{align}
where each user is assumed to observe the same number of propagation paths, denoted by $N_p$, the strength associated with the $l$th path seen by the $k$th user is represented by $\alpha_l^k$ (assuming ${\alpha_l^k \sim \mathcal {CN}}(0, 1)$), and $\phi_l^k(\theta_l^k)$ is the random azimuth (elevation) angle of departure drawn independently from uniform distributions over $[0, 2\pi]$. ${\bf a}(\phi_l^k, \theta_l^k)$ is the array response vector depending only on array structures. Here we consider a uniform linear array (ULA) whose array response vector admits a simple expression, given by \cite[Eq.~(6)]{Heath:mmWaveSparse} where $d$ is the normalized antenna spacing.

We compare in Fig.~\ref{fig:mmWaveCompareScheme} our proposed PZF scheme against the beamspace MIMO (B-MIMO) scheme proposed in \cite{Sayeed:MUmmWave}, which essentially steers streams onto the approximate strongest paths (using DFT matrix columns) at the RF domain and performs low-dimensional baseband ZF precoding based on the equivalent channel. For fair comparison, the BS is also assumed to have a total of $K$ chains. The B-MIMO scheme achieves desirable performance in line-of-sight (LoS) channel but fails to capture sparse multipath components in non-LoS channels.

\begin{figure}
\centering
\begin{minipage}{.46\textwidth}
\centering
\includegraphics[width = \textwidth, clip,keepaspectratio]{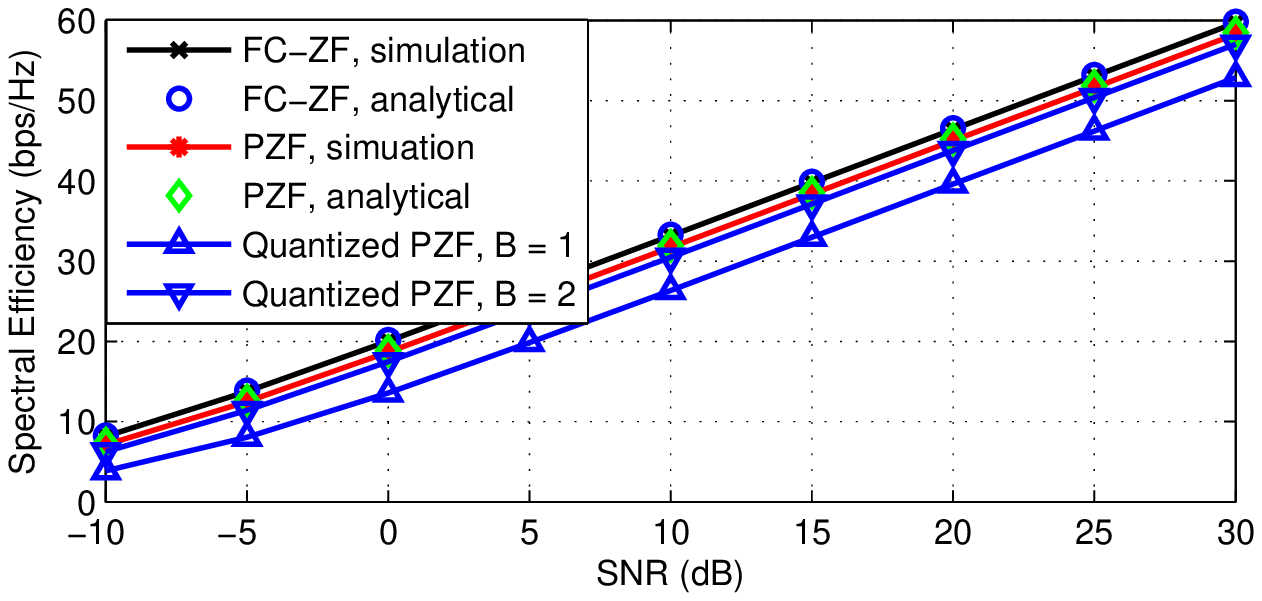}
\caption{Spectral efficiency achieved by different precoding schemes in large mmWave multiuser systems with $N_t = 128, K = 4, d = \frac{1}{2}$ and $N_p = 10$, obtained from averaging $1000$ channel realizations.} \label{fig:mainCompareScheme}
\end{minipage}\hfill
\begin{minipage}{.46\textwidth}
\centering
\includegraphics[width = \textwidth, clip,keepaspectratio]{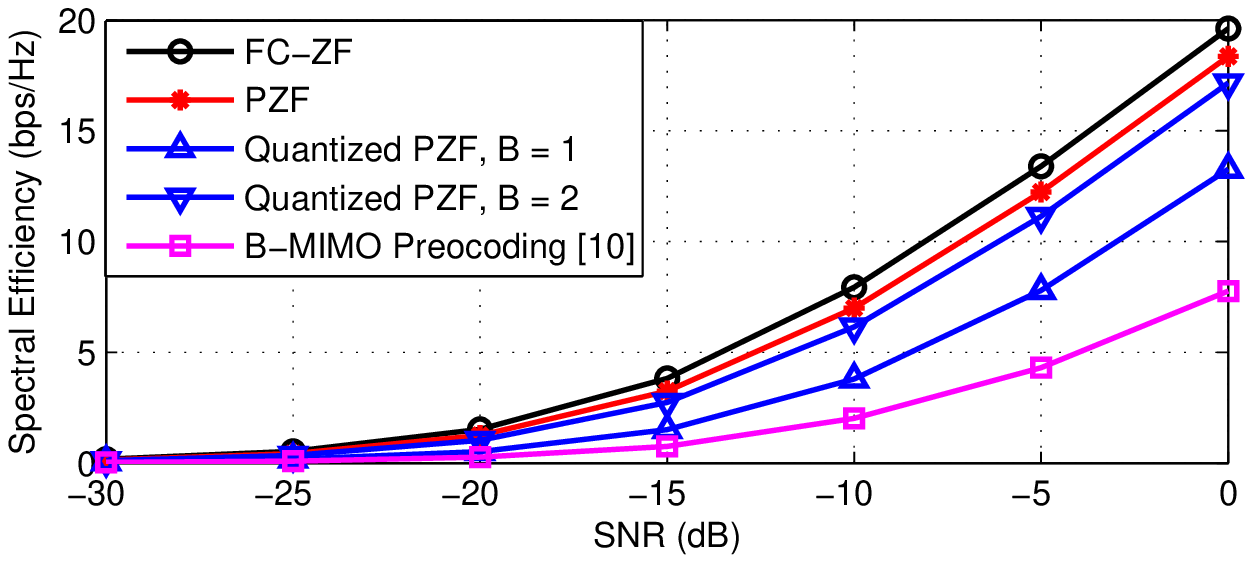}
\caption{Spectral efficiency achieved by different precoding schemes in large mmWave multiuser systems with $N_t = 128, K = 4, d = \frac{1}{2}$ and $N_p = 10$, obtained from averaging $1000$ channel realizations.}\label{fig:mmWaveCompareScheme}
\end{minipage}
\end{figure}

\section{Conclusion}
In this paper, we have studied a large multiuser MIMO system under practical RF hardware constraints. We have proposed to approach the desirable yet infeasible full-complexity ZF precoding with low-complexity hybrid PZF scheme. The RF processing was designed to harvest the large power gain with reasonable complexity, and the baseband precoder was then introduced to facilitate multi-stream processing. Its performance has been characterized in a closed form and further demonstrated in both Rayleigh fading and poorly scattered mmWave channels through computer simulations.

\end{document}